\documentclass[]{interact}

\usepackage{amsmath}
\usepackage{epstopdf}

\usepackage[caption=false]{subfig}
\pdfoutput=1
\usepackage[numbers,sort&compress]{natbib}
\bibpunct[, ]{[}{]}{,}{n}{,}{,}
\makeatletter
\def\NAT@def@citea{\def\@citea{\NAT@separator}}
\makeatother

\theoremstyle{plain}
\newtheorem{theorem}{Theorem}[section]
\newtheorem{lemma}[theorem]{Lemma}
\newtheorem{corollary}[theorem]{Corollary}

\numberwithin{equation}{section}

\theoremstyle{definition}
\newtheorem{definition}[theorem]{Definition}

\theoremstyle{remark}

\begin{document}

\articletype{RESEARCH ARTICLE}

\title{Two generalizations of ideal matrices and their applications}

\author{
\name{Mingpei Zhang, Heng Guo{*} and Wenlin Huang \thanks{CONTACT{*} Heng Guo. Email: guoheng@ruc.edu.cn}}
\affil{School of Mathematics, Renmin University of China, Beijing, China}
}

\maketitle

\begin{abstract}
In this paper, two kinds of  generalizations of ideal matrices, generalized ideal matrices and double ideal matrices, are obtained and studied. The concepts of generalized ideal matrices and double ideal matrices are proposed, and their ranks and maximal linearly independent groups are verified.

The initial motivation to study double cyclic matrices is to study the quasi cyclic codes of the fractional index. In this paper, the generalized form of the quasi cyclic codes, i.e. the $\phi$-quasi cyclic codes, and the construction of the generated matrix are given by the double ideal matrix. 
\end{abstract}

\begin{keywords}
Ideal matrix ; eigenvector ; $\phi$-cyclic code ; coding theory ; generated matrix
\end{keywords}

MSC codes: 15A03

\section{Introduction}

Circulant matrices are a kind of matrices with good properties and special structure, and it is widely used in the industrial field and other branches of mathematics. Therefore, the research based on classical circulant matrices has also become an important part of matrix theory \cite{1}. And, many variants of circulant matrices have also emerged, such as g-circulant matrices \cite{1}, r-circulant matrices \cite{2}, RFMLR-circulant matrices \cite{3}, etc.

Among the many generalized forms of circulant matrices, ideal matrices \cite{4,5}, generalized circulant matrices and double circulant matrices \cite{6}, which are studied for industrial applications, have many valuable properties. The structure and many properties of the ideal matrix are similar to the circulant matrix, and the double circulant matrix is the union of the matrix formed by two vectors after several cyclic shifts. The properties of both make them widely used in coding theory, especially in the theory of cyclic codes. Ideal matrices are mainly used in the construction of ideal lattices and $\phi$-cyclic codes and in the improvement of the NTRU cryptosystem \cite{7}, double The cyclic matrix is used in the quasi cyclic code. Therefore, in this paper, inspired by the above content, we give the definition of the generalized ideal matrix and the double ideal matrix and study the properties of both, and propose $\phi$-quasi cyclic code.

In \cite{8}, the rank of the circulant matrix is calculated. Considering the application of the matrix in coding theory, it is necessary to evaluate the maximum linearly independent group of the matrix. In Section \ref{section2}, refer to The eigenvector \cite{9} of the circulant matrix, after obtaining the form of the generalized ideal matrix and the double ideal matrix, some conclusions related to them are obtained, such as rank, maximum linear independent group and some other properties. It also provides the basis for the proofs of some theorems in the next section.

Circular matrices play an important role in the study of cyclic codes and their generalization. A preliminary study of cyclic codes was presented in \cite{10}, and the structure of $\phi$-cyclic codes was described in \cite{4}. In Section \ref{section3}, a class of $\phi$-quasi cyclic codes is obtained based on the related properties of double ideal matrices. For the related study of the proposed cyclic codes, see \cite{11,12,13}.

\section{Ideal matrix and its generalizations}\label{section2}
\subsection{Ideal matrix}
Let the polynomial
\begin{equation}\label{equation2.1}
\phi(x)=x^n-\phi_{n-1}x^{n-1}-\cdots-\phi_{1}x-\phi_{0}\in\mathbb{Z}[x],\phi_{0}\neq 0.
\end{equation}
be a polynomial with no multiple roots over the complex field $\mathbb{C}$. And let $w_{1},w_{2},\cdots,w_{n}$ be $n$ distinct non-zero roots of $\phi(x)$. Then according to the parameters $\phi_{0},\phi_{1},\cdots,\phi_{n-1}$ given in (\ref{equation2.1}), the rotation matrix \cite{5} can be given by
\begin{equation*}
H=H_{\phi}=\begin{bmatrix}
    \begin{matrix}0&\cdots&0\ \end{matrix}
    &\phi_{0}\\
    \Large{I_{n-1}}&\begin{matrix}
    \phi_{1}\\ \vdots\\ \phi_{n-1}
    \end{matrix}
\end{bmatrix}_{n\times n}\in\mathbb{Z}^{n\times n}.
\end{equation*}
$I_{n-1}$ is an $n-1$-order identity matrix, and the characteristic polynomial of the matrix is $\phi(x)$. In particular, when $\phi_{0}=1$ and $\phi_{1}=\cdots=\phi_{n-1}=0$, $H_{\phi}$ is an $n$-order basic circulant matrix.

Next, the rotation matrix $H_{\phi}$ is applied to the column vector in $\mathbb{R}^{n}$. Let $f=\begin{bmatrix}f_{0}\\f_{1}\\ \vdots\\ f_{n-1}\end{bmatrix}\in\mathbb{R}^{n}$, the ideal matrix can be defined as 
\begin{equation}\label{equation2.2}
H^{*}(f)=[f,Hf,H^{2}f,\cdots,H^{n-1}f]_{n\times n}\in\mathbb{R}^{n\times n}.
\end{equation}

It is obvious that ideal matrix $H^{*}(f)$ is a generalization of circulant matrix. If $\phi(x)=x^{n}-1$, then $H^{*}(f)$ is the ordinary circulant matrix; If $\phi(x)=x^{n}-r$, then $H^{*}(f)$ is the $r$-circulant matrix. For the related properties of ideal matrix, please see Theorem 2 in \cite{4}.

\subsection{Generalized ideal matrix}
According to the ideal matrix given in (\ref{equation2.2}), a more general form of matrix can be given by
\begin{definition}\label{definition2.1}
For rotation matrix $H_{\phi}$ and $n$-dimensional real column vector $f$, the $n\times m$ matrix can be constructed by 
\begin{equation*}
H^{*}(f)_{n\times m}=[f,Hf,H^{2}f,\cdots,H^{m-1}f]_{n\times m}.
\end{equation*}
This matrix is called a $n\times m$ generalized ideal matrix.
\end{definition}

In particular, the matrix is a conventional ideal matrix when $m=n$, and a generalized circulant matrix when $H$ is a basic circulant matrix. For the properties of generalized circulant matrices, see Theorem 2.4 of \cite{6}. In this paper, we will give some lemmas that will be used to prove the properties of generalized ideal matrices.

First of all, for the rotation matrix $H_{\phi}$, 
\begin{equation*}
H_{\phi}^{T}=\begin{bmatrix}
    \begin{matrix}0\\ \vdots\\0 \end{matrix}
    &\Large{I_{n-1}}\\
    \phi_{0} &\begin{matrix}
    \phi_{1}& \cdots & \phi_{n-1}
    \end{matrix}
\end{bmatrix}_{n\times n},
\end{equation*}
the transpose has the same pairwise different eigenvalues $w_{1},w_{2},\cdots,w_{n}$. Therefore we have the following lemma.
\begin{lemma}\label{lemma2.2}
The eigenvalue $w_{i}$ of $H_{\phi}^{T}$ has an eigenvector $\begin{bmatrix}1\\w_{i}\\ \vdots\\ w_{i}^{n-1}\end{bmatrix}(1\leq i\leq n).$
\end{lemma}
\begin{proof}
The matrix calculation gives
$$
\begin{aligned}
H_{\phi}^{T}\begin{bmatrix}1\\w_{i}\\ \vdots\\ w_{i}^{n-1}\end{bmatrix}&=\begin{bmatrix}
    \begin{matrix}0\\ \vdots\\0 \end{matrix}
    &\Large{I_{n-1}}\\
    \phi_{0} &\begin{matrix}
    \phi_{1}& \cdots & \phi_{n-1}
    \end{matrix}
\end{bmatrix}\begin{bmatrix}1\\w_{i}\\ \vdots\\ w_{i}^{n-1}\end{bmatrix}\\
&=\begin{bmatrix}w_{i}\\w_{i}^{2}\\ \vdots\\ \phi_{0}+\phi_{1}w_{i}+\cdots+\phi_{n-1}w_{i}^{n-1}\end{bmatrix}.
\end{aligned}
$$
Since the characteristic polynomial of $H_{\phi}^{T}$ is
$$\phi(x)=x^n-\phi_{n-1}x^{n-1}-\cdots-\phi_{1}x-\phi_{0},$$
Therefore
$$\phi_{0}+\phi_{1}w_{i}+\cdots+\phi_{n-1}w_{i}^{n-1}=w_{i}^{n},$$
So there is
$$H_{\phi}^{T}\begin{bmatrix}1\\w_{i}\\ \vdots\\ w_{i}^{n-1}\end{bmatrix}=\begin{bmatrix}w_{i}\\w_{i}^{2}\\ \vdots\\ w_{i}^{n}\end{bmatrix}=w_{i}\begin{bmatrix}1\\w_{i}\\ \vdots\\ w_{i}^{n-1}\end{bmatrix}.$$
\end{proof}
Further conclusions can be drawn from Lemma \ref{lemma2.2}.
\begin{lemma}\label{lemma2.3}
For each root $w_{i}(1\leq i\leq n)$ of $\phi(x)$ and an $n$-dimensional real column vector $f$ and its corresponding polynomial $f(x)=\sum_{j=0}^{n-1}f_{j}x^{j}$, the following equation holds.
\begin{equation*}
[H^{*}(f)_{n\times m}]^{T}\begin{bmatrix}1\\w_{i}\\ \vdots\\ w_{i}^{n-1}\end{bmatrix}=f(w_{i})\begin{bmatrix}1\\w_{i}\\ \vdots\\ w_{i}^{m-1}\end{bmatrix}.
\end{equation*}
\end{lemma}
\begin{proof}
$[H^{*}(f)_{n\times m}]^{T}$ can be written as a block matrix $\begin{bmatrix}f^{T}\\f^{T}H_{\phi}^{T}\\ \vdots\\ f^{T}(H_{\phi}^{m-1})^{T}\end{bmatrix}$, so we have the following equation.
$$
\begin{aligned}
[][H^{*}(f)_{n\times m}]^{T}\begin{bmatrix}1\\w_{i}\\ \vdots\\ w_{i}^{n-1}\end{bmatrix}&=\begin{bmatrix}f^{T}\\f^{T}H_{\phi}^{T}\\ \vdots\\ f^{T}(H_{\phi}^{m-1})^{T}\end{bmatrix}\begin{bmatrix}1\\w_{i}\\ \vdots\\ w_{i}^{n-1}\end{bmatrix}\\
&=\begin{bmatrix}f^{T}[1,w_{i},\cdots,w_{i}^{n-1}]^{T}\\f^{T}H_{\phi}^{T}[1,w_{i},\cdots,w_{i}^{n-1}]^{T}\\ \vdots\\ f^{T}(H_{\phi}^{m-1})^{T}[1,w_{i},\cdots,w_{i}^{n-1}]^{T}\end{bmatrix}.
\end{aligned}
$$
And then, from Lemma \ref{lemma2.2}, we can see
$$
\begin{aligned}
\begin{bmatrix}f^{T}[1,w_{i},\cdots,w_{i}^{n-1}]^{T}\\f^{T}H_{\phi}^{T}[1,w_{i},\cdots,w_{i}^{n-1}]^{T}\\ \vdots\\ f^{T}(H_{\phi}^{m-1})^{T}[1,w_{i},\cdots,w_{i}^{n-1}]^{T}\end{bmatrix}&=\begin{bmatrix}f^{T}[1,w_{i},\cdots,w_{i}^{n-1}]^{T}\\w_{i}f^{T}[1,w_{i},\cdots,w_{i}^{n-1}]^{T}\\ \vdots\\ w_{i}^{m-1}f^{T}[1,w_{i},\cdots,w_{i}^{n-1}]^{T}\end{bmatrix}\\
&=f^{T}\begin{bmatrix}1\\w_{i}\\ \vdots\\ w_{i}^{n-1}\end{bmatrix}\begin{bmatrix}1\\w_{i}\\ \vdots\\ w_{i}^{m-1}\end{bmatrix},
\end{aligned}
$$
and
$$f^{T}\begin{bmatrix}1\\w_{i}\\ \vdots\\ w_{i}^{n-1}\end{bmatrix}=f_{0}+f_{1}w_{i}+\cdots+f_{n-1}w_{i}^{n-1}=f(w_{i}),$$
so
$$[H^{*}(f)_{n\times m}]^{T}\begin{bmatrix}1\\w_{i}\\ \vdots\\ w_{i}^{n-1}\end{bmatrix}=f(w_{i})\begin{bmatrix}1\\w_{i}\\ \vdots\\ w_{i}^{m-1}\end{bmatrix}$$
holds.
\end{proof}
An $m\times n$ generalized Vandermonde matrix is then given by
\begin{equation*}
V_{\phi}^{m\times n}=\begin{bmatrix}
    1 & 1 & 1 & \dots & 1 \\
    w_{1} & w_{2} & w_{3} & \dots & w_{n} \\
    w_{1}^2 & w_{2}^2 & w_{3}^2 & \dots & w_{n}^2 \\
    \vdots & \vdots & \vdots & \ddots & \vdots \\
    w_{1}^{m-1} & w_{2}^{m-1} & w_{3}^{m-1} & \dots & w_{n}^{m-1}
\end{bmatrix}.
\end{equation*}
In particular, when $m=n$, it is a conventional $n\times n$ Vandermonde matrix, denoted by $V_{\phi}^{n}$. From this definition we obtain the following lemma.
\begin{lemma}\label{lemma2.4}
For any $n$-dimensional real vector $f$ and its corresponding polynomial  $f(x)=\sum_{j=0}^{n-1}f_{j}x^{j}$, the following equation
holds.
\begin{equation*}
[H^{*}(f)_{n\times m}]^{T}V_{\phi}^{n}=V_{\phi}^{m\times n}diag\{f(w_{1}),\cdots,f(w_{n})\}.
\end{equation*}
\end{lemma}
\begin{proof}
The proof can be obtained by Lemma \ref{lemma2.3}.
\end{proof}
Next we give the theorem of this subsection.
\begin{theorem}\label{theorem2.5}
Let $d(x)$ be $gcd(f(x),\phi(x))$, $d$ be $degd(x)$ and let $r=min\{m,$
$n-d\}$, Then:

(i)the rank of $H^{*}(f)_{n\times m}$ is $r$;

(ii)the first $r$ columns of $H^{*}(f)_{n\times m}$ are linearly independent.
\end{theorem}
\begin{proof}
Let all roots of $d(x)$ be $w_{j_{1}},\cdots,w_{j_{d}}$. These are also roots of $\phi(x)$, and it is known from the premise that the rest of the roots of $\phi(x)$ are not roots of $f(x)$. Here we denote this part of the roots as $w_{j_{d+1}},\cdots,w_{j_{n}}$. Since the $i$-th column of $V_{\phi}^{m\times n}diag\{f(w_{1}),\cdots,f(w_{n})\}$ can be expressed as
$$f(w_{i})\begin{bmatrix}1 \\ w_{i} \\ \vdots \\ w_{i}^{m-1}\end{bmatrix}(1\leq i\leq n),$$
it can be seen that there are $n-d$ non-zero columns.
$$f(w_{j})\begin{bmatrix}1 \\ w_{j_{t}} \\ \vdots \\ w_{j_{t}}^{m-1}\end{bmatrix}(d+1\leq t\leq n).$$

If $m\geq n-d$, the rank of $V_{\phi}^{m\times n}diag\{f(w_{1}),\cdots,f(w_{n})\}$ is $n-d$; 
 
 If $m<n-d$, the rank of $V_{\phi}^{m\times n}diag\{f(w_{1}),\cdots,f(w_{n})\}$ is $m$.

 Then the rank of the matrix can be written as $r$. By Lemma \ref{lemma2.4}, we can get that the rank of  $[H^{*}(f)_{n\times m}]^{T}V_{\phi}^{n}$ is also $r$. Since $V_{\phi}^{n}$ is an invertible matrix, it is easy to know that the rank of $[H^{*}(f)_{n\times m}]^{T}$ is $r$, i.e. the rank of  $H^{*}(f)_{n\times m}$ is $r$. 

Since a $r$-order non-zero minor determinant can be found in the first $r$ rows of  the matrix $V_{\phi}^{m\times n}diag\{f(w_{1}),\cdots,f(w_{n})\}$, the first $r$ rows of this matrix are linearly independent. By Lemma \ref{lemma2.4}, the first $r$ rows of $[H^{*}(f)_{n\times m}]^{T}V_{\phi}^{n}$ are linearly independent. Since $V_{\phi}^{n}$ is an invertible matrix and right multiplication of an invertible matrix does not change the linear independence of the rows, the first $r$ rows of $[H^{*}(f)_{n\times m}]^{T}$ are linearly independent, i.e. the first $r$ columns of $H^{*}(f)_{n\times m}$ are linearly independent.
\end{proof}
\begin{corollary}\label{corollary2.6}
$H^{*}(Hf)_{n\times m}$ also has the two properties described in Theorem \ref{theorem2.5}.
\end{corollary}
\begin{proof}
The rank of $H^{*}(f)_{n\times m}$ is $r$ and $H$ is an invertible matrix, so the rank of
\begin{equation*}
H^{*}(Hf)_{n\times m}=[Hf,H^{2}f,\cdots,H^{m}f]_{n\times m}=H\cdot H^{*}(f)_{n\times m}
\end{equation*}
is also $r$; $H\cdot H^{*}(f)_{n\times m}$ is equivalent to a series of primitive row transformations of $H^{*}(f)_{n\times m}$. Therefore, it does not affect the linear independence between the columns, so that the first $r$ columns of $H^{*}(Hf)_{n\times m}$ are also linearly independent.
\end{proof}
\begin{corollary}\label{corollary2.7}
Any $r$ consecutive columns of $H^{*}(f)_{n\times m}$ are linearly independent.
\end{corollary}
\begin{proof}
 The $d+1$-st to $d+r$-st columns of $H^{*}(f)_{n\times m}$, i.e. $$H^{d}f,\cdots,H^{d+r-1}f(0\leq d\leq m-r),$$ are in fact the first $r$ columns of $H^{*}(H^{d}f)_{n\times m}$. It follows step by step from Corollary \ref{corollary2.6} that $H^{*}(H^{d}f)_{n\times m}$ also has the two properties described in Theorem \ref{theorem2.5}, so that $H^{d}f,\cdots,H^{d+r-1}f$ are linearly independent.
\end{proof}

\subsection{Double ideal matrix}
\begin{definition}\label{definition2.8}
For $n_{1}\times n_{1}$ rotation matrix $H_{\phi_{1}}$, $n_{2}\times n_{2}$ rotation matrix $H_{\phi_{2}}$ $n_{1}$-dimensional real column vector $f_{1}$ and $n_{2}$-dimensional real column vector $f_{2}$, the $(n_{1}+n_{2})\times m$ matrix can be constructed by
\begin{equation*}
\begin{aligned}
H_{\phi_{1},\phi_{2}}^{*}(f_{1},f_{2})_{(n_{1}+n_{2})\times m}&=\begin{bmatrix}H_{\phi_{1}}^{*}(f_{1})_{n_{1}\times m}\\H_{\phi_{2}}^{*}(f_{2})_{n_{2}\times m} \end{bmatrix} \\ &=\begin{bmatrix}f_{1}&H_{\phi_{1}}f_{1}&H_{\phi_{1}}^{2}f_{1}&\cdots&H_{\phi_{1}}^{m-1}f_{1}\\
f_{2}&H_{\phi_{2}}f_{2}&H_{\phi_{2}}^{2}f_{2}&\cdots&H_{\phi_{2}}^{m-1}f_{2}
\end{bmatrix}_{(n_{1}+n_{2})\times m}.
\end{aligned}
\end{equation*}
This matrix is called a $(n_{1}+n_{2})\times m$ generalized ideal matrix.
\end{definition}
In particular, the matrix is a double circulant matrix when $H_{\phi_{1}}$ and $H_{\phi_{2}}$ are both basic circulant matrices.trix. For the properties of double circulant matrices, see 
Theorem 3.6 of \cite{6}. Next, by comparison with Lemma \ref{lemma2.3} and Lemma \ref{lemma2.4}, two lemmas relating to double ideal matrices can be written.
\begin{lemma}\label{lemma2.9}
For each root $w_{i}(1\leq i\leq n_{1})$ of $\phi_{1}(x)$, each root $v_{j}(1\leq j\leq n_{2})$ of $\phi_{2}(x)$, an $n_{1}$-dimensional real column vector $f_{1}$, an $n_{2}$-dimensional real column vector $f_{2}$ and their corresponding polynomials $f_{1}(x)=\sum_{j=0}^{n_{1}-1}f_{1j}x^{j}$, 
$f_{2}(x)=\sum_{j=0}^{n_{2}-1}f_{2j}x^{j}$, the following equation holds.
\begin{equation*}
\begin{aligned}
& [H_{\phi_{1},\phi_{2}}^{*}(f_{1},f_{2})_{(n_{1}+n_{2})\times m}]^{T}[1,w_{i},\cdots,w_{i}^{n_{1}-1},0, \cdots,0]^{T} \\
&=f_{1}(w_{i})[1,w_{i},\cdots,w_{i}^{n_{1}-1},0, \cdots,0]^{T};
\end{aligned}
\end{equation*}
\begin{equation*}
\begin{aligned}
& [H_{\phi_{1},\phi_{2}}^{*}(f_{1},f_{2})_{(n_{1}+n_{2})\times m}]^{T}[0,\cdots,0,1,v_{j}, \cdots,v_{j}^{n_{2}-1}]^{T} \\
&=f_{2}(v_{j})[0,\cdots,0,1,v_{j}, \cdots,v_{j}^{n_{2}-1}]^{T}.
\end{aligned}
\end{equation*}
\end{lemma}
Construct a block matrix
\begin{equation*}
V_{\phi_{1},\phi_{2}}^{m\times (n_{1}+n_{2})}=[V_{\phi_{1}}^{m\times n_{1}},V_{\phi_{2}}^{m\times n_{2}}].
\end{equation*}
to get the next Lemma.
\begin{lemma}\label{lemma2.10}
For an $n_{1}$-dimensional real column vector $f_{1}$, an $n_{2}$-dimensional real column vector $f_{2}$ and their corresponding polynomials $f_{1}(x)=\sum_{j=0}^{n_{1}-1}f_{1j}x^{j}$, 
$f_{2}(x)=\sum_{j=0}^{n_{2}-1}f_{2j}x^{j}$, the following equation holds.
\begin{equation*}
\begin{aligned}
&[H_{\phi_{1},\phi_{2}}^{*}(f_{1},f_{2})_{(n_{1}+n_{2})\times m}]^{T}\begin{bmatrix}V_{\phi_{1}}^{n_{1}} & \\ & V_{\phi_{2}}^{n_{2}} \end{bmatrix}\\
&=V_{\phi_{1},\phi_{2}}^{m\times (n_{1}+n_{2})}diag\{f_{1}(w_{1}),\cdots,f_{1}(w_{n_{1}}),f_{2}(v_{1}),\cdots,f_{2}(v_{n_{2}})\}.
\end{aligned}
\end{equation*}
\end{lemma}
Next, similar to Theorem \ref{theorem2.5}, a theorem about double ideal matrices is given.
\begin{theorem}\label{theorem2.11}
Let $\phi_{3}(x)$ be $gcd(\phi_{1}(x),\phi_{2}(x))$, $n_{3}$ be $deg\phi_{3}(x)$. $d$ is the degree of
$$\frac{gcd(f_{1}(x),\phi_{1}(x))\cdot gcd(f_{2}(x),\phi_{2}(x))\cdot\phi_{3}(x)}{gcd(f_{1}(x)\cdot f_{2}(x),\phi_{3}(x))},$$
$r=min\{m,n_{1}+n_{2}-d\}$, the following conclusion is then reached.

(i)the rank of $H_{\phi_{1},\phi_{2}}^{*}(f_{1},f_{2})_{(n_{1}+n_{2})\times m}$ is $r$.

(ii)the first $r$ columns of $H_{\phi_{1},\phi_{2}}^{*}(f_{1},f_{2})_{(n_{1}+n_{2})\times m}$ are linearly independent.
\end{theorem}
\begin{proof}
To prove that the rank of 
$$V_{\phi_{1},\phi_{2}}^{m\times (n_{1}+n_{2})}diag\{f_{1}(w_{1}),\cdots,f_{1}(w_{n_{1}}),f_{2}(v_{1}),\cdots,f_{2}(v_{n_{2}})\}$$
is $r$ and that the first $r$ rows is linearly independent, the following symbols are given: 

$e_{1}$ is the degree of $gcd(f_{1}(x),\phi_{1}(x))$; 

$e_{2}$ is the degree of $gcd(f_{2}(x),\phi_{2}(x))$;

$e$ is the degree of $\frac{\phi_{3}(x)}{gcd(f_{1}(x)\cdot f_{2}(x),\phi_{3}(x))}$.

So, $d=e_{1}+e_{2}+e$.
Since $\phi_{3}(x)$ is a divisor of $\phi_{1}(x)$ and $\phi_{1}(x)$ has no multiple roots, $\phi_{3}(x)$ also has no multiple roots. Let all roots of $\phi_{3}(x)$ be $c_{1},\cdots,c_{n_{3}}$.

All roots of $gcd(f_{1}(x)\cdot f_{2}(x),\phi_{3}(x))$ are all common roots of $\phi_{3}(x)$ and $\phi_{1}(x)$ and all common roots of $\phi_{3}(x)$ and $\phi_{2}(x)$, so all roots of $\frac{\phi_{3}(x)}{gcd(f_{1}(x)\cdot f_{2}(x),\phi_{3}(x))}$ are the roots of $\phi_{3}(x)$ which are neither the roots of $f_{1}(x)$ nor the roots of $f_{2}(x)$. In other words, there are exactly $e$ roots in $\phi_{3}(x)$, denoted here as $$c_{j_{1}},\cdots,c_{j_{e}}(1\leq j_{1}<\cdots<j_{e}\leq n_{3})$$ such that $$f_{1}(c_{j_{i}})\not=0\not=f_{2}(c_{j_{i}})(i=1,\cdots,e).$$

Since $e_{1}$ is the degree of $gcd(f_{1}(x),\phi_{1}(x))$, there are exactly $e_{1}$ roots of $\phi_{1}(x)$ that are also roots of $f_{1}(x)$, and exactly $n_{1}-e_{1}$ roots that are not roots of $f_{1}(x)$, i.e. $f_{1}(w_{1}),\cdots,f_{1}(w_{n_{1}})$ has $e_{1}$ values that are zero and $n_{1}-e_{1}$ values that are not zero. Similarly, $f_{2}(v_{1}),\cdots,f_{2}(v_{n_{2}})$ has $e_{2}$ values which are zero and $n_{2}-e_{2}$ values which are not zero.

In summary, the following conclusions can be drawn:

There are $n_{1}-e_{1}$ roots in $\phi_{1}(x)$ that are not roots of $f_{1}(x)$, and these roots can be written as
$$w_{k_{1}},\cdots,w_{k_{n_{1}-e_{1}-e}},c_{j_{1}},\cdots,c_{j_{e}}(1\leq k_{1}<\cdots<k_{n_{1}-e_{1}-e}\leq n_{1});$$

There are $n_{2}-e_{2}$ roots in $\phi_{2}(x)$ that are not roots of $f_{2}(x)$, and these roots can be written as
$$v_{l_{1}},\cdots,v_{l_{n_{2}-e_{2}-e}},c_{j_{1}},\cdots,c_{j_{e}}(1\leq l_{1}<\cdots<l_{n_{2}-e_{2}-e}\leq n_{2}).$$

$w_{k_{1}},\cdots,w_{k_{n_{1}-e_{1}-e}},c_{j_{1}},\cdots,c_{j_{e}},v_{l_{1}},\cdots,v_{l_{n_{2}-e_{2}-e}}$ are different. So, in $$V_{\phi_{1},\phi_{2}}^{m\times (n_{1}+n_{2})}diag\{f_{1}(w_{1}),\cdots,f_{1}(w_{n_{1}}),f_{2}(v_{1}),\cdots,f_{2}(v_{n_{2}})\},$$ the columns represented by 
$$f_{1}(w_{k_{1}}),\cdots,f_{1}(w_{k_{n_{1}-e_{1}-e}}),f_{1}(c_{j_{1}}),\cdots,f_{1}(c_{j_{e}}),$$ $$f_{2}(v_{l_{1}}),\cdots,f_{2}(v_{l_{n_{2}-e_{2}-e}}),f_{2}(c_{j_{1}}),\cdots,f_{2}(c_{j_{e}})$$
are all non-zero columns, the rest are all zero vectors. The columns represented by $f_{1}(c_{j_{i}})$ and $f_{2}(c_{j_{i}})$ can be expressed linearly with each other $(1\leq i\leq e)$. Therefore, by the columns represented by 
$$f_{1}(w_{k_{1}}),\cdots,f_{1}(w_{k_{n_{1}-e_{1}-e}}),f_{1}(c_{j_{1}}),\cdots,f_{1}(c_{j_{e}}),f_{2}(v_{l_{1}}),\cdots,f_{2}(v_{l_{n_{2}-e_{2}-e}}),$$
it can be seen that

When $$m\geq n_{1}-e_{1}+n_{2}-e_{2}-e=n_{1}+n_{2}-d,$$ the rank of $V_{\phi_{1},\phi_{2}}^{m\times (n_{1}+n_{2})}diag\{f_{1}(w_{1}),\cdots,f_{1}(w_{n_{1}}),f_{2}(v_{1}),\cdots,f_{2}(v_{n_{2}})\}$ is $n_{1}+n_{2}-d$; 

When $$m<n_{1}-e_{1}+n_{2}-e_{2}-e=n_{1}+n_{2}-d,$$ the rank of $V_{\phi_{1},\phi_{2}}^{m\times (n_{1}+n_{2})}diag\{f_{1}(w_{1}),\cdots,f_{1}(w_{n_{1}}),f_{2}(v_{1}),\cdots,f_{2}(v_{n_{2}})\}$ is $m$.

So, the rank of this matrix is $r$. thus, the rank of $$[H_{\phi_{1},\phi_{2}}^{*}(f_{1},f_{2})_{(n_{1}+n_{2})\times m}]^{T}\begin{bmatrix}V_{\phi_{1}}^{n_{1}} & \\ & V_{\phi_{2}}^{n_{2}} \end{bmatrix}$$ is $r$. Since $\begin{bmatrix}V_{\phi_{1}}^{n_{1}} & \\ & V_{\phi_{2}}^{n_{2}} \end{bmatrix}$ is a full-rank matrix, the rank of $[H_{\phi_{1},\phi_{2}}^{*}(f_{1},f_{2})_{(n_{1}+n_{2})\times m}]^{T}$ is also $r$.

In the first $r$ rows of $$V_{\phi_{1},\phi_{2}}^{m\times (n_{1}+n_{2})}diag\{f_{1}(w_{1}),\cdots,f_{1}(w_{n_{1}}),f_{2}(v_{1}),\cdots,f_{2}(v_{n_{2}})\},$$ an $r$-order non-zero minor determinant can be found, so the first $r$ rows of this matrix are linearly independent. By Lemma \ref{lemma2.10}, we can get that the first $r$ rows of  $$[H_{\phi_{1},\phi_{2}}^{*}(f_{1},f_{2})_{(n_{1}+n_{2})\times m}]^{T}\begin{bmatrix}V_{\phi_{1}}^{n_{1}} & \\ & V_{\phi_{2}}^{n_{2}} \end{bmatrix}$$ are linearly independent. Since $\begin{bmatrix}V_{\phi_{1}}^{n_{1}} & \\ & V_{\phi_{2}}^{n_{2}} \end{bmatrix}$ is an invertible matrix and right multiplication of an invertible matrix does not change the linear independence of the rows, the first $r$ rows of $[H_{\phi_{1},\phi_{2}}^{*}(f_{1},f_{2})_{(n_{1}+n_{2})\times m}]^{T}$ are linearly independent, i.e. the first $r$ columns of $H_{\phi_{1},\phi_{2}}^{*}(f_{1},f_{2})_{(n_{1}+n_{2})\times m}$ are linearly independent.
\end{proof}
\begin{corollary}\label{corollary2.12}
$H_{\phi_{1},\phi_{2}}^{*}(H_{\phi_{1}}f_{1},H_{\phi_{2}}f_{2})_{(n_{1}+n_{2})\times m}$ also has the two properties described in Theorem \ref{theorem2.11}.
\end{corollary}
\begin{proof}
The rank of $H_{\phi_{1},\phi_{2}}^{*}(f_{1},f_{2})_{(n_{1}+n_{2})\times m}$ is $r$, $H_{\phi_{1}}$ and $H_{\phi_{2}}$ are full-rank matrices, so $\begin{bmatrix}H_{\phi_{1}}& \\ &H_{\phi_{2}}\end{bmatrix}$ is also a full-rank matrix. As a result, the rank of
\begin{equation*}
\begin{aligned}
&H_{\phi_{1},\phi_{2}}^{*}(H_{\phi_{1}}f_{1},H_{\phi_{2}}f_{2})_{(n_{1}+n_{2})\times m}\\
&=\begin{bmatrix}H_{\phi_{1}}^{*}(H_{\phi_{1}}f_{1})_{n_{1}\times m}\\H_{\phi_{2}}^{*}(H_{\phi_{1}}f_{2})_{n_{2}\times m} \end{bmatrix}\\
&=\begin{bmatrix}H_{\phi_{1}}f_{1}&H_{\phi_{1}}^{2}f_{1}&H_{\phi_{1}}^{3}f_{1}&\cdots&H_{\phi_{1}}^{m}f_{1}\\
H_{\phi_{2}}f_{2}&H_{\phi_{2}}^{2}f_{2}&H_{\phi_{2}}^{3}f_{2}&\cdots&H_{\phi_{2}}^{m}f_{2}
\end{bmatrix}_{(n_{1}+n_{2})\times m}\\
&=\begin{bmatrix}H_{\phi_{1}}& \\ &H_{\phi_{2}}\end{bmatrix}\begin{bmatrix}f_{1}&H_{\phi_{1}}f_{1}&H_{\phi_{1}}^{2}f_{1}&\cdots&H_{\phi_{1}}^{m-1}f_{1}\\
f_{2}&H_{\phi_{2}}f_{2}&H_{\phi_{2}}^{2}f_{2}&\cdots&H_{\phi_{2}}^{m-1}f_{2}
\end{bmatrix}_{(n_{1}+n_{2})\times m}
\end{aligned}
\end{equation*}
is also $r$. Therefore, the first $r$ columns of $H_{\phi_{1},\phi_{2}}^{*}(H_{\phi_{1}}f_{1},H_{\phi_{2}}f_{2})_{(n_{1}+n_{2})\times m}$ are linearly independent.
\end{proof}
\begin{corollary}\label{corollary2.13}
Any $r$ consecutive columns of $H_{\phi_{1},\phi_{2}}^{*}(f_{1},f_{2})_{(n_{1}+n_{2})\times m}$ are linearly independent.
\end{corollary}
\begin{proof}
The $d+1$-st to $d+r$-st columns of $H_{\phi_{1},\phi_{2}}^{*}(f_{1},f_{2})_{(n_{1}+n_{2})\times m}$ are in fact the first r columns of $H_{\phi_{1},\phi_{2}}^{*}(H_{\phi_{1}}^{d}f_{1},H_{\phi_{2}}^{d}f_{2})_{(n_{1}+n_{2})\times m}$.  It follows step by step from
Corollary \ref{corollary2.12} that $H_{\phi_{1},\phi_{2}}^{*}(H_{\phi_{1}}^{d}f_{1},H_{\phi_{2}}^{d}f_{2})_{(n_{1}+n_{2})\times m}$ also has the two properties described in Theorem \ref{theorem2.11}.
\end{proof}
At the end of this section we discuss a special case, namely $n_{1}+n_{2}=m$. In this case, the corresponding double ideal matrix is an $m\times m$ square matrix. For this type of matrix, the following corollary can be obtained.
\begin{corollary}\label{corollary2.14}
For the matrix $[H_{\phi_{1},\phi_{2}}^{*}(f_{1},f_{2})_{m\times m}]^{T}$, the following three kinds of eigenvectors with eigenvalues of 0 can be found.

(i)The vector corresponding to each root $w_{p}$ of $gcd(f_{1}(x),\phi_{1}(x))$
\begin{equation*}
[1,w_{p},\cdots,w_{p}^{n_{1}-1},0, \cdots,0]^{T};
\end{equation*}

(ii)The vector corresponding to each root $v_{q}$ of $gcd(f_{2}(x),\phi_{2}(x))$
\begin{equation*}
[0,\cdots,0,1,v_{q}, \cdots,v_{q}^{n_{2}-1}]^{T};
\end{equation*}

(iii)For $c_{j_{1}},\cdots,c_{j_{e}}$, mentioned in the proof of Theorem \ref{theorem2.11}, there are corresponding vectors
\begin{equation*}
[-f_{2}(c_{j}),-f_{2}(c_{j})c_{j},\cdots,-f_{2}(c_{j})c_{j}^{n_{1}-1},f_{1}(c_{j}),f_{1}(c_{j})c_{j}, \cdots,f_{1}(c_{j})c_{j}^{n_{2}-1}]^{T}.
\end{equation*}
\end{corollary}
\begin{proof}
Both (i) and (ii) follow directly from Lemma \ref{lemma2.9}, and
\begin{equation*}
\begin{aligned}
&[H_{\phi_{1},\phi_{2}}^{*}(f_{1},f_{2})_{m\times m}]^{T}\times \\
&[-f_{2}(c_{j}),-f_{2}(c_{j})c_{j},\cdots,-f_{2}(c_{j})c_{j}^{n_{1}-1},f_{1}(c_{j}),f_{1}(c_{j})c_{j}, \cdots,f_{1}(c_{j})c_{j}^{n_{2}-1}]^{T}\\
&=([H_{\phi_{1}}^{*}(f_{1})_{n_{1}\times m}]^{T},[H_{\phi_{2}}^{*}(f_{2})_{n_{2}\times m}]^{T})\times \\
&[-f_{2}(c_{j}),-f_{2}(c_{j})c_{j},\cdots,-f_{2}(c_{j})c_{j}^{n_{1}-1},f_{1}(c_{j}),f_{1}(c_{j})c_{j}, \cdots,f_{1}(c_{j})c_{j}^{n_{2}-1}]^{T}\\
&=-f_{2}(c_{j})f_{1}(c_{j})[1,c_{j},\cdots,c_{j}^{m}]^{T}+f_{1}(c_{j})f_{2}(c_{j})[1,c_{j},\cdots,c_{j}^{m}]^{T} \\
&=0.
\end{aligned}
\end{equation*}
So, (iii) is also true.
\end{proof}
\begin{corollary}\label{corollary2.15}
$H_{\phi_{1},\phi_{2}}^{*}(f_{1},f_{2})_{(n_{1}+n_{2})\times (n_{1}+n_{2})}$  is a full rank matrix if and only if $(f_{1}(x),\phi_{1}(x))=1$, $(f_{2}(x),\phi_{2}(x))=1$ and $(\phi_{1}(x),\phi_{2}(x))=1$.
\end{corollary}
\begin{proof}
According to Theorem \ref{theorem2.11}, the rank of a double ideal matrix is $n_{1}+n_{2}-d$. And from the proof of Theorem \ref{theorem2.11}, $d=e_{1}+e_{2}+e$, so $e_{1}=e_{2}=e=0$ if the double ideal square is full-rank. Therefore, we can first state that $(f_{1}(x),\phi_{1}(x))=1$ and $(f_{2}(x),\phi_{2}(x))=1$. if $e=0$, $$gcd(f_{1}(x)\cdot f_{2}(x),\phi_{3}(x))=\phi_{3}(x).$$ If $deg\phi_{3}(x)> 0$, then $(f_{1}(x)\cdot f_{2}(x),\phi_{3}(x))\neq 1$. So at least one of $(f_{1}(x),\phi_{1}(x))\neq 1$ and $(f_{2}(x),\phi_{2}(x))\neq 1$ holds. In this case, $d$ must be greater than $0$, so $deg\phi_{3}(x)=0$, i.e. $(\phi_{1}(x),\phi_{2}(x))=1$.
\end{proof}

\section{$\phi$-quasi cyclic codes}\label{section3}
The preconditions are given first: $q$ is a prime power and $\phi_{1}(x)$ is a $k$-degree polynomial al with no multiple roots over $\mathbb{F}_{q}$, denoted as
\begin{equation*}
\phi_{1}(x)=\phi_{1,0}+\phi_{1,1}x+\cdots+\phi_{1,k-1}x^{k-1}+x^{k}.
\end{equation*}
$\phi_{2}(x)$ is a $l$-degree polynomial over $\mathbb{F}_{q}$, denoted as
\begin{equation*}
\phi_{2}(x)=\phi_{2,0}+\phi_{2,1}x+\cdots+\phi_{2,l-1}x^{l-1}+x^{l}.
\end{equation*}
Let the degree of $\phi_{3}(x)=(\phi_{1}(x),\phi_{2}(x))$ be $m$, denoted as
\begin{equation*}
\phi_{3}(x)=\phi_{3,0}+\phi_{3,1}x+\cdots+\phi_{3,m-1}x^{m-1}+x^{m}.
\end{equation*}
Other provisions are as follows:

$\overline{R}_{k}=\mathbb{F}_{q}[x]/<\phi_{1}(x)>;$

$\overline{R}_{l}=\mathbb{F}_{q}[x]/<\phi_{2}(x)>;$

$\overline{R}_{k+l-m}=\mathbb{F}_{q}[x]/<\frac{\phi_{1}(x)\phi_{2}(x)}{\phi_{3}(x)}>$.

Next, consider the Cartesian product
\begin{equation*}
\overline{R}_{k}\times \overline{R}_{l}=\mathbb{F}_{q}[x]/<\phi_{1}(x)>\times \mathbb{F}_{q}[x]/<\phi_{2}(x)>.
\end{equation*}
In $\overline{R}_{k}\times \overline{R}_{l}$, each element can be expressed in the form $(\overline{a}(x),\overline{b}(x))$, where 
\begin{equation*}
\begin{aligned}
\overline{a}(x)=a_{0}+a_{1}x+\cdots +a_{k-1}x^{k-1}\in\overline{R}_{k};\\
\overline{b}(x)=b_{0}+b_{1}x+\cdots +b_{l-1}x^{l-1}\in\overline{R}_{l}.
\end{aligned}
\end{equation*}
Therefore, $(\overline{a}(x),\overline{b}(x))$ forms a one-to-one correspondence with the following codewords:
\begin{equation*}
(a_{0},a_{1},\cdots,a_{k-1},b_{0},b_{1},\cdots,b_{l-1})\in\mathbb{F}_{q}^{k}\times\mathbb{F}_{q}^{l}.
\end{equation*}
At this point, try multiplying the codeword $(\overline{a}(x),\overline{b}(x))$ by an $x$, we have
\begin{equation*}
\begin{aligned}
x(\overline{a}(x),\overline{b}(x))
&=x(a_{0}+a_{1}x+\cdots a_{k-1}x^{k-1},b_{0}+b_{1}x+\cdots +b_{l-1}x^{l-1})\\
& \equiv(\phi_{1,0}a_{k-1}+(a_{0}+\phi_{1,1}a_{k-1})x+\cdots +(a_{k-2}+\phi_{1,k-1}a_{k-1})x^{k-1}\\
& \pmod{\phi_{1}(x)},\\
& \phi_{2,0}b_{l-1}+(b_{0}+\phi_{2,1}b_{l-1})x+\cdots+(b_{l-2}+\phi_{2,l-1}b_{l-1})x^{l-1}\\
& \pmod{\phi_{2}(x)}).
\end{aligned}
\end{equation*}
Thus, based on the above, the definition of a $\phi$-quasi cyclic code is as follows
\begin{definition}\label{definition3.1}
For the linear subspace $\overline{C}$ of $\overline{R}_{k}\times \overline{R}_{l}$, if $x(c_{1}(x),c_{2}(x))\in\overline{C}$ for any $(c_{1}(x),c_{2}(x))\in \overline{C}$, then $\overline{C}$ is called a $\phi$-quasi cyclic code over $\overline{R}_{k}\times \overline{R}_{l}$.
\end{definition}
If $(\overline{a}(x),\overline{b}(x))\in \overline{R}_{k}\times \overline{R}_{l}$, then
\begin{equation*}
f(x)(\overline{a}(x),\overline{b}(x))=(f(x)\overline{a}(x)\pmod{\phi_{1}(x)},f(x)\overline{b}(x)\pmod{\phi_{2}(x)})\in\overline{R}_{k}\times \overline{R}_{l}.
\end{equation*}
The $\phi$-quasi cyclic code $\overline{C}_{\overline{a},\overline{b}}$ can be written in the following form:
\begin{equation*}
\{(f(x)\overline{a}(x)\pmod{\phi_{1}(x)},f(x)\overline{b}(x)\pmod{\phi_{2}(x)})\in\overline{R}_{k}\times \overline{R}_{l}|f(x)\in\overline{R}_{k+l-m}\}.
\end{equation*}
where $(\overline{a}(x),\overline{b}(x))$ is called the generating element of the $\phi$-quasi cyclic code.Then the ideal of $\overline{R}_{k+l-m}$ generated by $f(x)$ is written as $<f(x)>_{\overline{R}_{k+l-m}}$. 

In this paper we give a theorem to describe the related properties of $\phi$-quasi cyclic codes.
\begin{theorem}\label{theorem3.2}
For each $(\overline{a}(x),\overline{b}(x))\in \overline{R}_{k}\times \overline{R}_{l}$, let
\begin{equation*}
\begin{aligned}
\overline{g}_{\overline{a},\overline{b}}(x)&=gcd(\overline{a}(x),\frac{\phi_{1}(x)}{\phi_{3}(x)})\times gcd(\overline{b}(x),\frac{\phi_{2}(x)}{\phi_{3}(x)})\times gcd(\overline{a}(x),\overline{b}(x),\phi_{3}(x)).\\
\overline{h}_{\overline{a},\overline{b}}(x)&=\frac{\phi_{1}(x)\phi_{2}(x)}{\phi_{3}(x)\overline{g}_{\overline{a},\overline{b}}(x)}.
\end{aligned}
\end{equation*}
And give the mapping $\overline{r}_{\overline{a},\overline{b}}$:
\begin{equation*}
f(x)\in\overline{R}_{k+l-m}\rightarrow (f(x)\overline{a}(x)\pmod{\phi_{1}(x)},f(x)\overline{b}(x)\pmod{\phi_{2}(x)})\in\overline{R}_{k}\times \overline{R}_{l}.
\end{equation*}
The following conclusions are drawn:

(i)$im(\overline{r}_{\overline{a},\overline{b}})=\overline{C}_{\overline{a},\overline{b}}$;

(ii)$ker(\overline{r}_{\overline{a},\overline{b}})=<\overline{h}_{\overline{a},\overline{b}}(x)>_{\overline{R}_{k+l-m}}$, where $dim\overline{C}_{\overline{a},\overline{b}}=deg\overline{h}_{\overline{a},\overline{b}}(x)$;

(iii)$\overline{r}_{\overline{a},\overline{b}}$ can induce an isomorphism: $<\overline{g}_{\overline{a},\overline{b}}(x)>_{\overline{R}_{k+l-m}}\rightarrow\overline{C}_{\overline{a},\overline{b}}$, so $\overline{C}_{\overline{a},\overline{b}}$ can be written as $\{(f(x)\overline{a}(x)\pmod{\phi_{1}(x)},f(x)\overline{b}(x)\pmod{\phi_{2}(x)})\in\overline{R}_{k}\times \overline{R}_{l}
|f(x)\in<\overline{g}_{\overline{a},\overline{b}}(x)>_{\overline{R}_{k+l-m}}\}$.
\end{theorem}
\begin{proof}
From the definition of $\overline{C}_{\overline{a},\overline{b}}$ it is easy to see that (i) holds.

$f(x)\in ker(\overline{r}_{\overline{a},\overline{b}})$ if and only if
$$\left\{
\begin{aligned}
f(x)\overline{a}(x)\equiv 0\pmod{\phi_{1}(x)}\\
f(x)\overline{b}(x)\equiv 0\pmod{\phi_{2}(x)}
\end{aligned}
\right.$$
It is important to note that
$$\left\{
\begin{aligned}
\phi_{1}(x)=\phi_{3}(x)\frac{\phi_{1}(x)}{\phi_{3}(x)}\\
gcd(\frac{\phi_{1}(x)}{\phi_{3}(x)},\phi_{3}(x))=1
\end{aligned}
\right.$$
and
$$\left\{
\begin{aligned}
\phi_{2}(x)=\phi_{3}(x)\frac{\phi_{2}(x)}{\phi_{3}(x)}\\
gcd(\frac{\phi_{2}(x)}{\phi_{3}(x)},\phi_{3}(x))=1
\end{aligned}
\right.$$
In summary, we will get
$$\left\{
\begin{aligned}
f(x)\overline{a}(x)\equiv 0\pmod{\frac{\phi_{1}(x)}{\phi_{3}(x)}}\\
f(x)\overline{a}(x)\equiv 0\pmod{\phi_{3}(x)}\\
f(x)\overline{b}(x)\equiv 0\pmod{\frac{\phi_{2}(x)}{\phi_{3}(x)}}\\
f(x)\overline{b}(x)\equiv 0\pmod{\phi_{3}(x)}
\end{aligned}
\right.$$
Then, by combining the second and fourth congruence equations, we obtain
$$\left\{
\begin{aligned}
f(x)\overline{a}(x)\equiv 0\pmod{\frac{\phi_{1}(x)}{\phi_{3}(x)}}\\
f(x)\overline{b}(x)\equiv 0\pmod{\frac{\phi_{2}(x)}{\phi_{3}(x)}}\\
f(x)gcd(\overline{a}(x),\overline{b}(x))\equiv 0\pmod{\phi_{3}(x)}
\end{aligned}
\right.$$
Continue sorting to get
$$\begin{aligned}
& f(x)\equiv 0\\
& \pmod{\frac{\phi_{1}(x)}{\phi_{3}(x)gcd(\overline{a}(x),\frac{\phi_{1}(x)}{\phi_{3}(x)})}\times\frac{\phi_{2}(x)}{\phi_{3}(x)gcd(\overline{b}(x),\frac{\phi_{2}(x)}{\phi_{3}(x)})}\times \frac{\phi_{3}(x)}{gcd(\overline{a}(x),\overline{b}(x),\phi_{3}(x))}},\end{aligned}$$
namely
\begin{equation*}
f(x)\equiv 0 \pmod{\frac{\phi_{1}(x)\phi_{2}(x)}{\phi_{3}(x)\overline{g}_{\overline{a},\overline{b}}(x)}=\overline{h}_{\overline{a},\overline{b}}(x)}.
\end{equation*}
So $ker(\overline{r}_{\overline{a},\overline{b}})=<\overline{h}_{\overline{a},\overline{b}}(x)>_{\overline{R}_{k+l-m}}$ holds.

Since
\begin{equation*}
\overline{R}_{k+l-m}=<\overline{h}_{\overline{a},\overline{b}}(x)>_{\overline{R}_{k+l-m}}\oplus<\overline{g}_{\overline{a},\overline{b}}(x)>_{\overline{R}_{k+l-m}}
\end{equation*}
and $ker(\overline{r}_{\overline{a},\overline{b}})=<\overline{h}_{\overline{a},\overline{b}}(x)>_{\overline{R}_{k+l-m}}$, (iii) also holds.
\end{proof}
Next, this paper will study the generating matrix of $\phi$-quasi cyclic codes. For the $k$-dimensional vector $\overline{a}=(a_{0},a_{1},\cdots,a_{k-1})$ corresponding to the polynomial $\overline{a}(x)=a_{0}+a_{1}x+\cdots +a_{k-1}x^{k-1}\in\overline{R}_{k}$, a $k\times k$ ideal matrix can be generated.
\begin{equation*}
H_{\phi_{1}}^{*}(\overline{a}^{T})=[\overline{a}^{T},H_{\phi_{1}}\overline{a}^{T},\cdots,H_{\phi_{1}}^{k-1}\overline{a}^{T}].
\end{equation*}
Transpose the matrix and record it as 
\begin{equation*}
A=[H_{\phi_{1}}^{*}(a^{T})]^{T}=\begin{bmatrix}
    \overline{a} \\
    \overline{a}H_{\phi_{1}}^{T} \\
    \vdots \\
    \overline{a}[H_{\phi_{1}}^{T}]^{k-1}
\end{bmatrix}.
\end{equation*}
Similarly, an $l\times l$ ideal matrix can be obtained by $\overline{b}=(b_{0},b_{1},\cdots,b_{l-1})$:
\begin{equation*}
B=[H_{\phi_{2}}^{*}(\overline{b}^{T})]^{T}=\begin{bmatrix}
    \overline{b} \\
    \overline{b}H_{\phi_{2}}^{T} \\
    \vdots \\
    \overline{b}[H_{\phi_{2}}^{T}]^{l-1}
\end{bmatrix}
\end{equation*}
Let the greatest common divisor of $k$ and $l$ be $t$, so we can use $A$ and $B$ to construct a $\frac{kl}{t}\times (k+l)$ block matrix.
\begin{equation}\label{equation3.1}
\overline{(A,B)}=\begin{bmatrix}
    A & B \\
    A[H_{\phi_{1}}^{T}]^{k} & B[H_{\phi_{2}}^{T}]^{l} \\
    \vdots & \vdots \\
    A[H_{\phi_{1}}^{T}]^{k(\frac{l}{t}-1)} & B[H_{\phi_{2}}^{T}]^{l(\frac{k}{t}-1)}
\end{bmatrix}
\end{equation}
In the matrix, the first column has a total of $\frac{l}{t}$ blocks, and the second column has a total of $\frac{k}{t}$ blocks, and both columns can be regarded as the transpose of a generalized ideal matrix, then the matrix as a whole can be regarded as a transpose of a double ideal matrix. It is easy to see that each row of the matrix corresponds to a codeword in the $\phi$-quasi cyclic code $\overline{C}_{\overline{a},\overline{b}}$, and that all rows of the matrix linearly express all codewords in $\overline{C}_{\overline{a},\overline{b}}$. From the above, the following conclusions can be drawn:
\begin{corollary}\label{corollary3.2}
Any $r=min\{\frac{kl}{t},k+l-d\}$ consecutive rows of the matrix $\overline{(A,B)}$ in (\ref{equation3.1}) can form a generating matrix of $\overline{C}_{\overline{a},\overline{b}}$, where $d$ is the degree of
$$\frac{gcd(\overline{a}(x),\phi_{1}(x))\cdot gcd(\overline{b}(x),\phi_{2}(x))\cdot\phi_{3}(x)}{gcd(\overline{a}(x)\cdot \overline{b}(x),\phi_{3}(x))}.$$
\end{corollary}
\begin{proof}
Since the matrix $\overline{(A,B)}$ is a transpose of a double ideal matrix, it follows from Theorem \ref{theorem2.11} that the rank of $\overline{(A,B)}$ is equal to $r=min\{\frac{kl}{t},k+l-d\}$. It follows from Corollary \ref{corollary2.13} that any $r$ consecutive rows of $\overline{(A,B)}$ are linearly independent and therefore form a generating matrix for $\overline{C}_{\overline{a},\overline{b}}$.
\end{proof}
The $\phi$-quasi cyclic codes mentioned in this section can be regarded as a generalization of \cite{14,15}.

\bibliographystyle{IEEEtran}

\end{document}